\documentclass{article}
\usepackage{spconf,amsmath,graphicx}

\usepackage{calc}
\usepackage{url}
\usepackage{latexsym}
\usepackage{amsfonts}
\usepackage{amsmath,amsthm}   
\usepackage{amssymb}
\usepackage{bm}
\usepackage{capitalgreekitalic}
\usepackage{balance}

\DeclareGraphicsExtensions{.pdf,.jpeg,.png,.jpg,.emf,.eps}

\newtheorem{lemma}{Lemma}

\newtheorem{remark}{Remark}

\def\U{{\bf U}}
\def\V{{\bf V}}
\def\F{{\bf F}}
\def\T{{\bf T}}
\def\G{{\bf G}}
\def\H{{\bf H}}
\def\U{{\bf U}}

\def\r{{\bf r}}
\def\s{{\bf s}}
\def\u{{\bf u}}
\def\Thetab{\bm{\Theta}}
\def\Sigmab{\bm{\Sigma}}

\def\tr{\operatorname{tr}}

\def\diag{\operatorname{diag}}

\title{Interference Leakage Minimization \\ in RIS-assisted MIMO Interference Channels}
\name{Ignacio~Santamaria$^1$, Mohammad Soleymani$^2$, Eduard Jorswieck$^3$, Jes{\'u}s Guti{\'e}rrez$^4$ \thanks{The work of I. Santamaria was supported under grant PID2019-104958RB-C43 (ADELE) funded by MCIN/AEI /10.13039/501100011033. The work of E. Jorswieck and J. Guti{\'e}rrez was supported in part by the Federal Ministry of Education and Research (BMBF, Germany) as part of the 6G Research and Innovation Cluster 6G-RIC under Grants 16KISK031 and 16KISK026, respectively.}}
\address{\normalsize$^1$Department of Communications Engineering, Universidad de Cantabria, 39005 Santander, Spain\\ \normalsize
$^2$Signal and System Theory Group, Universit{\"a}t  Paderborn, 33098 Paderborn, Germany\\ \normalsize
$^3$Institute for Communications Technology, Technische Universit{\"a}t Braunschweig, 38106 Braunschweig,
Germany\\ \normalsize
$^4$IHP - Leibniz-Institut
f{\"u}r Innovative Mikroelektronik, 15236 Frankfurt (Oder), Germany
}
\begin{document}
\ninept
\maketitle
\begin{abstract}
We address the problem of interference leakage (IL) minimization in the $K$-user multiple-input multiple-output (MIMO) interference channel (IC) assisted by a reconfigurable intelligent surface (RIS). We describe an iterative algorithm based on block coordinate descent to minimize the IL cost function. A reformulation of the problem provides a geometric interpretation and shows interesting connections with envelope precoding and phase-only zero-forcing beamforming problems. As a result of this analysis, we derive a set of necessary (but not sufficient) conditions for a phase-optimized RIS to be able to perfectly cancel the interference on the $K$-user MIMO IC.
\end{abstract}

\section{Introduction}
\label{sec:intro}

The ability of reconfigurable intelligent surfaces (RISs) to reflect signals that can superimpose coherently at the desired receivers to boost signal power, or can null interference at the unintended receivers, has resulted in a wide range of applications for RIS-assisted wireless communication systems \cite{{RenzoJSAC2020},{pan2020multicell},{ZapponeTWT2019},{ZhangTWT2019},{SoleymaniTVT22}}.

In this work, we consider interference channels (IC) assisted by RIS, which have attracted much attention recently. In \cite{SchoberTCOM2022} the authors study the degrees of freedom (DoF) of the $K$-user RIS-assisted single-input single-output (SISO) IC with channel or symbol extensions (time-varying channels). In \cite{HuangTVT2020} the authors characterize the achievable rate region of a multiple-input single-output (MISO) interference channel with $K$ users where each user is assisted by one RIS. In \cite{FuICC2021} the authors optimize the RIS to maximize the DoF of the $K$-user MIMO IC, but they consider symbol extensions. In \cite{AbradoCL2021} a weighted minimum mean square error (wMMSE) approach that accounts for the mutual coupling among the RIS elements is used to maximize the sum rate.

In this paper, we focus on the $K$-user multiple-input multiple-output (MIMO) interference channel (IC) without symbol extensions, where the RIS goal is to align, cancel, or neutralize the interference. As a cost function, we consider the total interference power at all receivers, a cost function which is known in the interference alignment literature as {\it interference leakage} (IL) \cite{gomadam2011distributed}. The IL can be minimized by a joint design of the precoders/decoders and the RIS, or by the RIS alone. We consider
{\it active RISs}, for which the amplitudes and phases can be independently optimized, and {\it passive lossless RISs}, for which only the phase shifts can be optimized \cite{{SchoberTCOM2022},{SchoberPoorTCOM2023}}. We will refer to the latter simply as RIS. We show that, under a mild condition on the number of RIS elements, an active RIS can always perfectly cancel the interference (zero-IL) and the problem has a simple closed-form solution. When only the phases can be optimized, it is necessary to apply some iterative algorithm to minimize IL. A reformulation of the IL minimization problem in RIS-assisted systems provides a geometric interpretation of the problem and shows interesting connections with envelope precoding and phase-only zero-forcing beamforming problems \cite{PanJTSP2014},\cite{ZhangTSP2016},\cite{ZhaoJTSP2016}. As a result of this analysis, we derive a set of necessary (but not sufficient) conditions for a phase-optimized RIS to be able to perfectly cancel the interference on the $K$-user MIMO IC.

\section{IL Minimization in RIS-Assisted Systems }
\label{sec:IL}
\subsection{System Model}
We consider a $K$-user MIMO interference channel assisted by a RIS that facilitates or enables interference alignment (IA) \cite{Jafar08} or interference neutralization \cite{Ho12}. The $k$th user has $T_k$ transmit antennas, $R_k$ receive antennas and transmits $d_k$ data streams. According to the commonly used notation, we denote the MIMO-IC in abbreviated form as $(T_k\times R_k, d_k)^K$. We focus on scenarios for which IA is not feasible without the assistance of a RIS, meaning that it is not possible to perfectly cancel the interference at all unintended receivers by optimizing only the precoders and the decoders \cite{gonzalez2014feasibility}. The equivalent MIMO channel from the $l$th transmitter to the $k$th receiver is
\begin{equation}
\widetilde{\H}_{lk}= \H_{lk} + \F_k^H \Thetab \G_l,
\label{eq:MIMOchannel}
\end{equation}
where $\H_{lk} \in \mathbb{C}^{R_k \times T_l}$ is the $(l,k)$ MIMO interference channel, $\G_l \in \mathbb{C}^{M \times T_l}$ is the channel from the $l$th transmitter to the RIS, $\F_k \in \mathbb{C}^{M \times R_k}$ is the channel from the RIS to the $k$th receiver, and $\bm{\Theta}=\text{diag} \left(\r \right)$ is the $M\times M$ diagonal RIS matrix with diagonal $ \r = \left(r_1, r_2,\ldots,r_M\right)^T$, where $|r_m|$ and ${\rm arg} (r_m) \in [-\pi, \pi)$ are the amplitude and phase shift of the $m$th reflecting element. 

\subsection{IL cost function}
The IL minimization problem is to find precoders $\V_l \in \mathbb{C}^{T_l\times d_l}$ for $l=1,\ldots,K$; decoders $\U_k \in \mathbb{C}^{R_k\times d_k}$ for $k=1,\ldots,K$; and RIS elements $\bm{\Theta} = {\rm diag}(\r)$ that minimize the IL
\begin{equation}
IL(\{\V_l\},\{\U_k\},\Thetab)=\sum_{l \neq k} \|\U_k^H \left( \H_{lk}+ \F_k^H \Thetab \G_l\right) \V_l  \|_F^2,
\label{eq:ILcostfunction}
\end{equation}
where each term of the sum is the squared Frobenius norm of the equivalent MIMO interference channel from the $l$th transmitter to the $k$th receiver in \eqref{eq:MIMOchannel} after precoding and decoding.

The minimization of \eqref{eq:ILcostfunction} may be carried out through a 3-step alternating optimization process, where in each step a set of variables (decoders, precoders, or RIS) is optimized while the other variables are held fixed:
\begin{enumerate}
    \item Optimize $\{\U_k\}_{k=1}^K$ for fixed $(\{\V_l\}_{l=1}^K, \Thetab)$.
    \item Optimize $\{\V_l\}_{l=1}^K$ for fixed $(\{\U_k\}_{k=1}^K, \Thetab)$.
    \item Optimize $\Thetab$ for fixed $(\{\U_k\}_{k=1}^K, \{\V_l\}_{l=1}^K)$.
    
\end{enumerate}

Steps 1 and 2, which obtain the precoders and decoders that minimize the IL while keeping the elements of the RIS fixed, can be solved by applying some of the existing methods \cite{gomadam2011distributed, Gonzalez14}. Therefore, in this paper, we mainly focus on RIS optimization for fixed precoders and decoders (Step 3 above). Let $\bar{\F}_k = \F_k \U_k$ be the $M\times d_k$ equivalent channel from the RIS to the receiver after decoding, let $\bar{\G}_l = \G_l \V_l$ be the $M\times d_l$ equivalent channel from the transmitter to the RIS after precoding, and let $\bar{\H}_{lk}=\U_k^H \H_{lk} \V_l$ be the $d_k \times d_l$ equivalent channel matrix after precoding-decoding. Therefore, the IL as a function solely of the RIS elements is  \cite{pan2020multicell}
\begin{eqnarray}
IL(\Thetab) = & \sum_{l \neq k} \| \bar{\H}_{lk} + \bar{\F}_k^H \Thetab \bar{\G}_l \|_F^2 \nonumber \\
= & \tr(\T) + \r^H \Sigmab \r + 2 {\rm Re}(\r^H {\bf s})
\label{eq:ILcostfunction1}
\end{eqnarray}
where $\T =  \sum_{l \neq k} \bar{\H}_{lk}^H \bar{\H}_{lk}$,  ${\bf s} = \sum_{l \neq k} {\rm diag}(\bar{\F}_k \bar{\H}_{lk}\bar{\G}_l^H) $, and $\Sigmab = \sum_{l \neq k} \bar{\F}_k\bar{\F}_k^H \odot \left(\bar{\G}_l \bar{\G}_l^H\right)^*$ (where $\odot$ denotes Hadamard product, $(\cdot)^H$ denotes Hermitian, and $(\cdot)^*$ denotes complex conjugate). It is clear that ${\bf Q}_{F_k}= \bar{\F}_k\bar{\F}_k^H$ and ${\bf Q}_{G_l} = \left(\bar{\G}_l \bar{\G}_l^H\right)^*$ are, respectively, rank-$d_k$ and rank-$d_l$ positive definite matrices (we assume $M > \max(d_k, d_l)$ $ \forall k,l$). Then, ${\bf Q}_{F_k} \odot {\bf Q}_{G_l}$ is also a positive semidefinite matrix of rank $d_k d_l$ \cite{Horn2020}, and $\Sigmab = \sum_{l \neq k} {\bf Q}_{F_k} \odot {\bf Q}_{G_l}^*$ has rank $g = \sum_{l \neq k} d_k d_l $.

Let us assume that $M > g$, which is not a restrictive condition since $g$ is usually much smaller than the number of RIS elements. Under this assumption, $\Sigmab$ is a rank-deficient semidefinite $M \times M$ matrix with eigendecomposition 
\begin{equation*}
    \Sigmab = \begin{bmatrix} \U_{\text{signal}} & \U_{\text {noise}} \end{bmatrix}
    \begin{bmatrix} \diag(\lambda_1,\ldots, \lambda_g) & {\bf 0} \\ {\bf 0} & {\bf 0}\end{bmatrix}
    \begin{bmatrix} \U_{\text{signal}}^H \\ \U_{\text {noise}}^H \end{bmatrix}.
\end{equation*}
We will refer to $\U_{\text{signal}} \in  \mathbb{C}^{M\times g}$ and $\U_{\text {noise}} \in  \mathbb{C}^{M\times (M-g)}$ as the {\em signal} and {\em noise} subspaces, respectively. Furthermore, we  may write $\Sigmab = \U_{\text{signal}} {\bm{\Lambda}} \U_{\text{signal}}^H$ where ${\bm {\Lambda}} = {\rm diag}(\lambda_1,\ldots,\lambda_g)$.

\subsection{Minimum IL RIS}
In this subsection, we consider the problem of optimizing the RIS elements to minimize IL. We consider active and passive lossless RISs. With an active RIS, it is possible to achieve zero interference zero and the IL minimization problem has a simple closed-form solution. When only the phases of the RIS can be optimized, it is necessary to apply some iterative algorithm. In particular, we describe one based on an alternating optimization procedure, which is guaranteed to converge to a stationary point.

\subsubsection{Active RIS}
For an active RIS, the IL can be completely canceled regardless of the precoders and the decoders under the assumption $M > g$. From the definition of ${\bf s} = \sum_{l \neq k} {\rm diag}(\bar{\F}_k \bar{\H}_{lk}\bar{\G}_l^H)$, it follows that ${\bf s} \in {\rm colspan}(\Sigmab)$, and therefore the optimal unconstrained solution for the RIS coefficients is 
\begin{equation*} 
\r_{unc} = -\Sigmab^{\sharp} \ \s = -  \U_{\text{signal}} {\bm{\Lambda}}^{-1} \U_{\text{signal}}^H \, {\bf s} =  \U_{\text{signal}} \, {\bm {\alpha}},
\end{equation*}
where $\Sigmab^{\sharp}$ denotes the pseudoinverse of $\Sigmab$ and we have defined the coordinates of $\r_{unc}$ in the signal subspace basis as ${\bm {\alpha}} = -{\bm{\Lambda}}^{-1}\U_{\text{signal}}^H \, {\bf s}$. Since $\Sigmab \Sigmab^{\sharp}$ is a rank-g projection matrix onto the signal subspace spanned by the columns of $\U_{\text{signal}}$ \cite[pp. 389]{Coherence}, and ${\bf s} \in {\rm colspan}(\Sigmab)$, it follows that 
\begin{equation*}
    2 {\rm Re} (\r^H \Sigmab  \r_{unc}) = -2 {\rm Re} (\r^H  \Sigmab \Sigmab^{\sharp} \ \s) = -2 {\rm Re} (\r^H \s).
\end{equation*}
Furthermore, it is easy to check that $\tr(\T) = \r_{unc}^H\Sigmab\r_{unc}$. Then, the IL cost function \eqref{eq:ILcostfunction1} may be rewritten as
\begin{equation*}
IL(\r) = (\r-\r_{unc})^H\Sigmab (\r-\r_{unc}),
\end{equation*}
and, clearly, the optimal unconstrained solution $\r_{unc}$ achieves $IL(\r_{unc}) = 0$.

\begin{remark}
Note that it is possible to enforce the constraint $|r_m| \leq 1$, $\forall m$, or the constraint $\| \r \|_2^2 \leq 1$ as a regularized version of the unconstrained solution $ - \U_{\text{signal}}  ({\bm{\Lambda} + \mu {\bf I}})^{-1}\U_{\text{signal}} ^H \,{\bf s}$, where $\mu \geq 0$ is a positive regularization parameter that may be chosen via bisection to enforce the required constraint. 
\end{remark}

\subsubsection{Passive RIS with $|r_m| = 1$}
The solution for a RIS satisfying $|r_m| = 1$, $\forall m$ can be found by solving the following problem
\begin{align}\label{eq_RIS_UMproblem1}
({\cal P}_1): \,\min_{\r}\,\,& (\r-\r_{unc})^H \Sigmab  (\r-\r_{unc}) \nonumber\\
\text{s.t.}\,\,&|r_m| =1, \forall m.
\end{align}
This is a unit-modulus quadratic programming problem. The function to be minimized is convex but the constraint is not. In the literature there are several algorithms  to solve quadratic programming problems like \eqref{eq_RIS_UMproblem1}  \cite{SidiropoulosTSP2012,pan2020multicell,tsinos2017efficient,ZhaoJTSP2016}. 
Next, we describe a block coordinate descent method that is computationally efficient. At each iteration, we fix all values of $\r$ except $r_m = e^{j \theta_m}$. Denoting $\bar{m} = \{ 1, \ldots, m-1,m+1,\ldots, M \}$, the IL as a function of $r_m$ can be written as
\begin{equation*}
IL(r_m) = C + 2 {\rm Re}(r_m^*(s_m + \Sigmab_{\bar{m}}^H \r_{\bar{m}})), 
\end{equation*}
where $C$ is a positive constant, $\Sigmab_{\bar{m}}$ denotes the $m$th column of $\Sigmab$ with the $m$th element removed, and 
\[\r_{\bar{m}} = \left( r_1,\, \, \ldots, \, \, r_{m-1}, \, \, r_{m+1}, \, \,\ldots, \, \, r_M \right)^T.
\]
The optimization problem is
\begin{align*}
\min_{r_m}\,\,&  {\rm Re}(r_m^*(s_m + \Sigmab_{\bar{m}}^H \r_{\bar{m}})) \\
\text{s.t.}\,\,& |r_m| = 1,
\end{align*}
which has the following closed-form solution
\begin{equation*}
\theta_m = \angle{\left(s_m + \Sigmab_{\bar{m}}^H \r_{\bar{m}} \right)} - \pi.
\end{equation*}
This algorithm belongs to the category of block coordinate descent (BCD) methods. Since at each BCD step, the problem is univariate (i.e., a single element of the RIS is updated at each step), and the minimizer is unique, then its convergence to a stationary point is guaranteed \cite{{Bertsekas99},{razaviyayn2013unified}}. 

\section{Zero-IL RIS}
\label{sec-ii-d}
In this section, we address the feasibility problem to achieve zero IL by using a passive RIS with $|r_m| = 1$ for arbitrary precoders and decoders. The following lemma, which summarizes the main technical contribution of this work, gives a set of necessary (but not sufficient) conditions  for the existence of RIS that perfectly cancels interference.

\begin{lemma}\label{Prop:LemmaRIS} 
Consider a RIS-assisted IC $(T_k\times R_k, d_k)^K$ with arbitrary precoders and decoders. The RIS has $M > g = \sum_{l \neq k} d_k d_l$ elements so that the unconstrained solution is $\r_{unc} = \U_{\text{signal}} \, {\bm {\alpha}}$. Then, a set of necessary conditions for the existence of a passive lossless or unit-modulus RIS that achieves zero IL is
\begin{equation*}
    \alpha_i \in {\cal D}_i, \quad i=1, \ldots, g,
\end{equation*}
where each of the ${\cal D}_i$ represents an annular region with known inner and outer radii.
\end{lemma}
\begin{proof}
Since $M>g$ any solution of ${\cal P}_1$ in (\ref{eq_RIS_UMproblem1}) such that
\begin{equation}
    \r = \r_{unc} + \U_{\text{noise}} \, {\bm {\beta}} = \U_{\text{signal}} \, {\bm {\alpha}} + \U_{\text{noise}} \, {\bm {\beta}}
    \label{eq:vectorr} 
\end{equation}
does not modify the value of the cost function and, therefore, is a solution that achieves zero-IL as well. In \eqref{eq:vectorr}, ${\bm {\alpha}}$ is a $g \times 1$ complex vector $ {\bm {\alpha}} = (\alpha_1, \ldots, \alpha_g)^T$,  and ${\bm {\beta}} = (\beta_1, \ldots, \beta_{M-g})^T$ is a $(M-g) \times 1$ complex vector.
The feasibility problem amounts to answer the following question: does there exist any $\r$, or, equivalently, any ${\bm {\beta}} \in \mathbb{C}^{(M-g)\times 1}$, such that the vector $ \U_{\text{signal}} \, {\bm {\alpha}} +  \U_{\text{noise}} \, {\bm {\beta}}$ has unit-modulus components? To try to answer this question, the problem can be reformulated as follows. Let $\r = (e^{j\theta_1},\ldots, e^{j\theta_M})^T$ be a vector with the RIS elements, and let $\U = [ \U_{\text{signal}} \,  \U_{\text{noise}}]=$ $ [\u_1,\ldots, \u_g,\u_{g+1},\ldots,\u_{M}]$ be a basis obtained from the eigendecomposition of the rank-deficient matrix $\Sigmab$. Premultiplying \eqref{eq:vectorr}  by any of the $g$ basis vectors of $ \U_{\text{signal}}$, the following system of nonlinear equations is formed
\begin{equation}
    \u_i^H\r  = \sum_{m=1}^M u_i^*(m)e^{j\theta_m}  = \alpha_i, \quad i= 1,\ldots,g, 
    \label{eq:polygonsabis} 
\end{equation}
where $\alpha_i= \u_i^H\r_{unc}$ for $i=1,\ldots,g$, are known complex values, while the RIS phases $\theta_m$ for $m=1,\ldots,M$, are unknown.
Note that if there exists a RIS satisfying \eqref{eq:polygonsabis}, then the  equations $ \U_{\text{noise}}^H \, \r ={\bm {\beta}}$ are automatically satisfied with $\|{\bm {\beta}} \|_2^2 = M- \|{\bm {\alpha}} \|_2^2$ due to the orthogonality between $ \U_{\text{signal}}$ and $ \U_{\text{noise}}$. Therefore, for fixed precoders and decoders, the RIS-assisted interference neutralization problem amounts to finding phases $\theta_m$ for $m=1,\ldots, M$ satisfying \eqref{eq:polygonsabis}.

For notational convenience, we denote $g_{i,m} = |u_i(m)|$, $\phi_m = \theta_m - \angle{u_i(m)}$ for $m=1, \ldots,M $. Without loss of generality, we assume that the moduli are sorted as $g_{i,1} \geq g_{2,m} \geq \ldots \geq g_{i,M}$. Then, the zero-IL RIS feasibility problem is to determine whether there exists a set of phases $\phi_m$, $m=1,\ldots, M$, satisfying
\begin{align}\label{eq_zeroILproblem}
({\cal P}_2): \, \sum_{m=1}^M g_{i,m}e^{j\phi_m} & = \alpha_i, \quad i= 1,\ldots,g.
\end{align}
Given a vector $\left(g_{i,1},\ldots, g_{i,M}\right)$ in \eqref{eq_zeroILproblem}, it is possible to characterize the region of the complex plane that can be reached by varying the phases of the RIS. More formally, this set is defined as
\begin{equation*}
    {\cal D}_i = \left \{  \sum_{m=1}^M g_{i,m}e^{j\phi_m} \, \bigg\vert \, \phi_m \in [0, 2 \pi), \,\, m=1, \ldots M \right \}.
\end{equation*}
If $\alpha_i \in {\cal D}_i$ then there are (possibly a continuum of) phases that satisfy $\sum_{m=1}^M g_{i,m}e^{j\phi_m}  = \alpha_i$. The region ${\cal D}_i$ was first studied in the context of a constant envelope precoding problem \cite{LarssonTWT2012}, where it was shown that ${\cal D}_i$ is a doughnut region given by
\begin{equation*}
    {\cal D}_i = \left \{  \alpha_i \in \mathbb{C} \quad | \quad  r_i \leq |\alpha_i| \leq R_i\right \},
\end{equation*}
where the outer radius is $R_i = \sum_{m=1}^M g_{i,m}$. An explicit expression for the radius of the inner circle was derived in \cite{WKMa_icassp13}. Theorem 1 in \cite{WKMa_icassp13} shows that $r_i = \max \{ g_{i,1} - \sum_{m=2}^M g_{i,m}, 0 \}$. If the difference is negative, the inner radius is zero and the doughnut region becomes a disk of radius $R_i$. As a direct application of these results, Lemma \ref{Prop:LemmaRIS} states that if any of the $\alpha_i$ does not belong to its corresponding region ${\cal D}_i$, then problem ${\cal P}_2$ in \eqref{eq_zeroILproblem} is infeasible. 
\end{proof}

\noindent From another perspective, each of the equations in \eqref{eq_zeroILproblem} can be interpreted as the equation of a polygon of known sides $(g_{i,1},\ldots, g_{i,M}$ $, |\alpha_i| )$ in the complex plane. The question of whether it is possible or not to form a polygon of given sides has been recently studied in the context of envelope precoding \cite{{PanJTSP2014},{ZhangTSP2016}} and secure communications through phase-only zero-forcing (ZF) beamforming \cite{ZhaoJTSP2016}. 
\noindent Therefore, it is possible to write an equivalent set of necessary conditions for the feasibility of \eqref{eq_zeroILproblem} based on whether it is possible to form a polygon with each of the equations of the problem.

\begin{remark} It may be illustrative to consider the simplest scenario with $M=2$ and $g=1$ so that $\Sigmab$ is a rank-one $2 \times 2$ matrix. The one-dimensional signal and noise subspaces of $\Sigmab$ are $\u_{\text{signal}}$ and $\u_{\text{noise}}$. This toy example could represent for example a SISO cognitive radio scenario assisted by a RIS with just $M=2$ elements located near the primary receiver that tries to cancel the interference produced by the secondary transmitter (there is a single interference stream and hence $g=1$). The unconstrained active RIS solution is $\r_{unc} = \u_{\text{signal}} \, \alpha$. In addition, according to the notation introduced in the paper $g_1 = |\u_{\text{signal}}(1)|$ and $g_2 = |\u_{\text{signal}}(2) |$. The results of this section specialize to $M=2$ as follows. A necessary (and in this case also sufficient) condition for the existence of a passive phase-only RIS that perfectly cancels the interference is that $\alpha$ belongs to the region
\begin{equation*}
    {\cal D} = \left \{  \alpha \in \mathbb{C} \quad | \quad   |g_1-g_2| \leq |\alpha| \leq g_1+g_2 \right \}.
\end{equation*}
Alternatively, this condition amounts to saying that it is possible to form a triangle with sides $(g_1,g_2,|\alpha|)$. Furthermore, it is easy to show that in this case there are exactly two solutions to problem \eqref{eq_RIS_UMproblem1} given by
\begin{equation*}
    \r = {\bf u}_s \alpha + {\bf u}_n \sqrt{2-|\alpha|^2} e^{j\theta_{\beta}},
\end{equation*}
where
\begin{equation*}
    \theta_{\beta} = \pm \cos^{-1} \left( \frac{1-|a|^2 - |b|^2}{2 |a||b|} \right) +\theta_{a} -\theta_{b},
\end{equation*}
where 
\begin{align*}
    |a|e^{j\theta_{a}} &= \alpha {\bf u}_{signal}(1),   \\
     |b|e^{j\theta_{b}} &= \sqrt{2-|\alpha|^2} {\bf u}_{noise}(1).
\end{align*}
\end{remark}
 
\section{Simulation results}
\begin{figure}[th]
    \centering
\includegraphics[width=.5\textwidth]{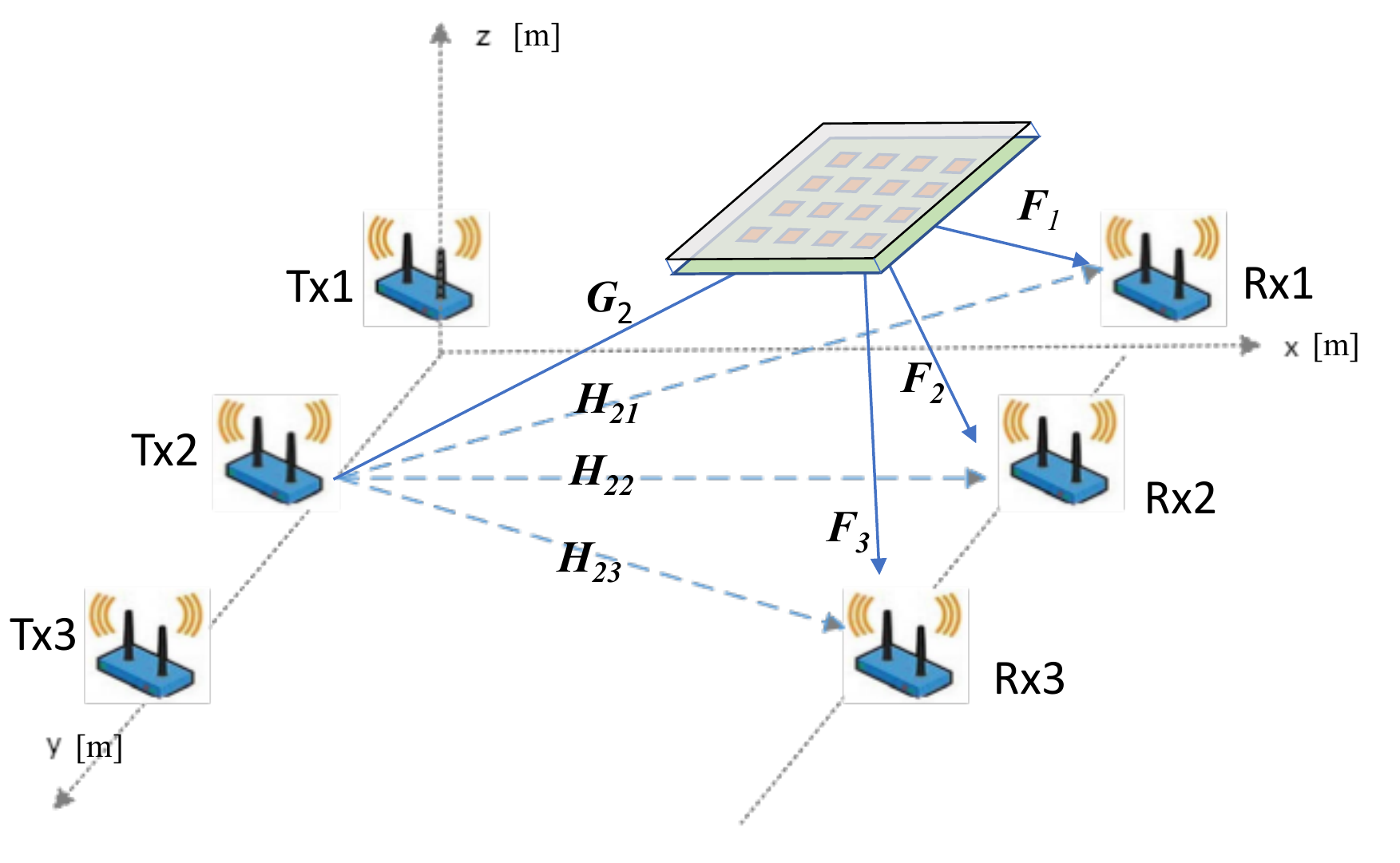}
     \caption{Simulation setup for the RIS-assisted MIMO IC.}
	\label{fig:setup}
\end{figure}
We consider a $(3 \times 3,2)^K$ MIMO-IC assisted by a RIS with $M$ elements for $K=3$ and $K=2$ users. IA is infeasible for this network without the assistance of a RIS, which is to say that it is not possible to achieve zero-IL by optimizing only the precoders and decoders \cite{gonzalez2014feasibility}. The coordinates $(x,y,z)$ in meters of the three transmitters are (0,\,0,\,2), (0,\,25,\,2), and (0,\,50,\,2); respectively. The receivers are located at (50,\,0,\,2), (50,\,25,\,2), and (50,\,50,\,2). There is only one RIS located at (40,\, 25,\,15) (cf. Fig. \ref{fig:setup}). For $K=2$ users, the link from Tx2 to Rx2 in Fig. \ref{fig:setup} is eliminated. The large-scale path loss in dB is given by
\begin{equation*}
PL = -30 - 10 \beta \log_{10} \left( d \right),
\label{eq:pathloss}
\end{equation*}
where $d$ is the link distance, and $\beta$ is the path-loss exponent. The direct Tx-Rx links, $\H_{lk}$, are assumed to be non-line-of-sight (NLOS) channels, with path-loss exponent $\beta = 3.75$ and small-scale Rayleigh fading. The Tx-RIS-Rx links, $\F_k$ and $\G_l$, are assumed to be line-of-sight (LOS) channels, with path-loss exponent $\beta = 2$ and small-scale Rice fading with Rician factor $\gamma =3$. A more detailed description of the system parameters can be found in \cite{soleymani2022improper}.
\begin{figure}[t]
    \centering
\includegraphics[width=.5\textwidth]{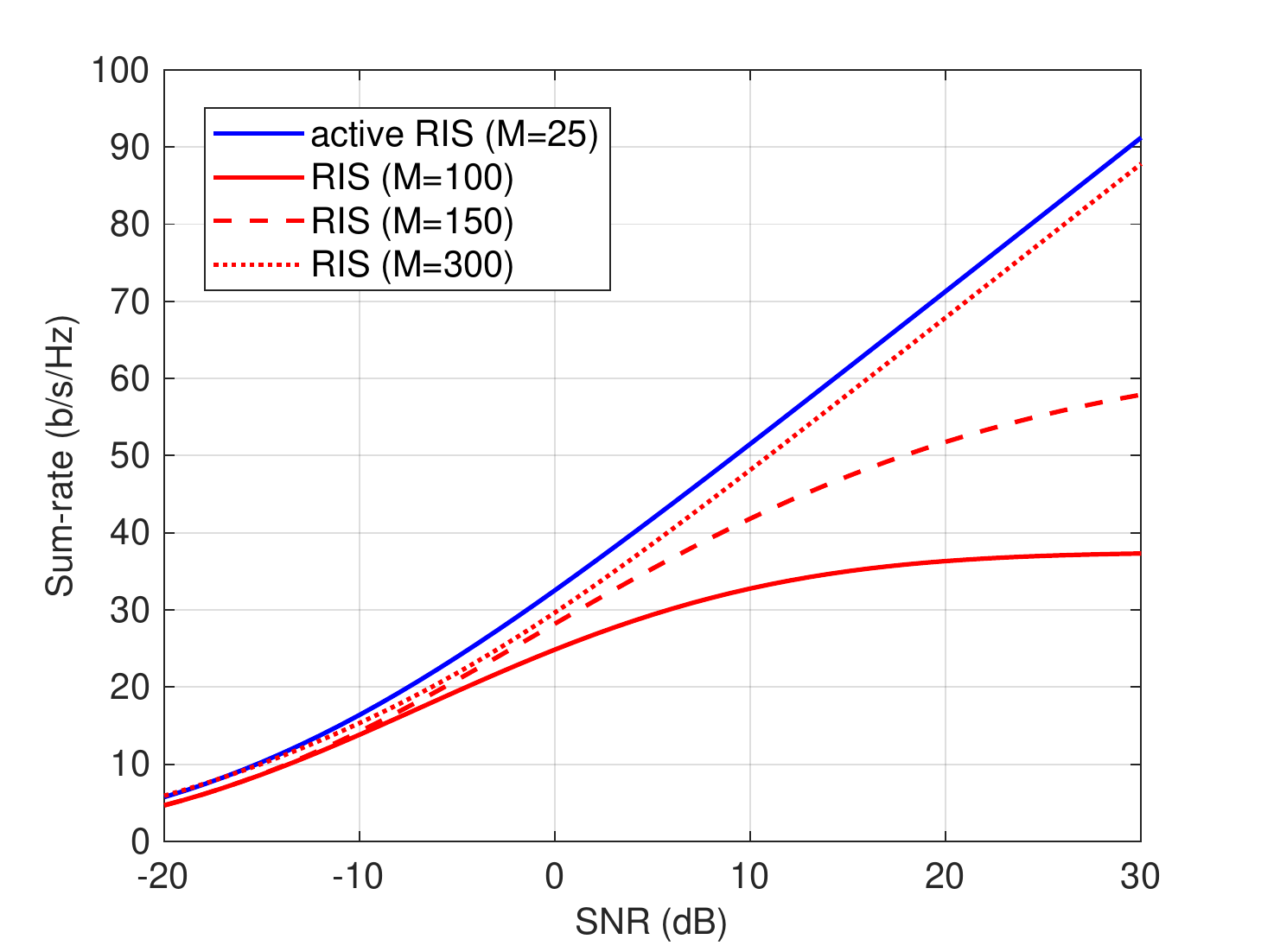}
     \caption{Sum-rate vs SNR for a RIS and an active RIS with different numbers of elements in a $(3 \times 3,2)^K$ MIMO IC.}
	\label{fig:sum-rate}
\end{figure}
\begin{figure}[ht]
    \centering
\includegraphics[width=.5\textwidth]{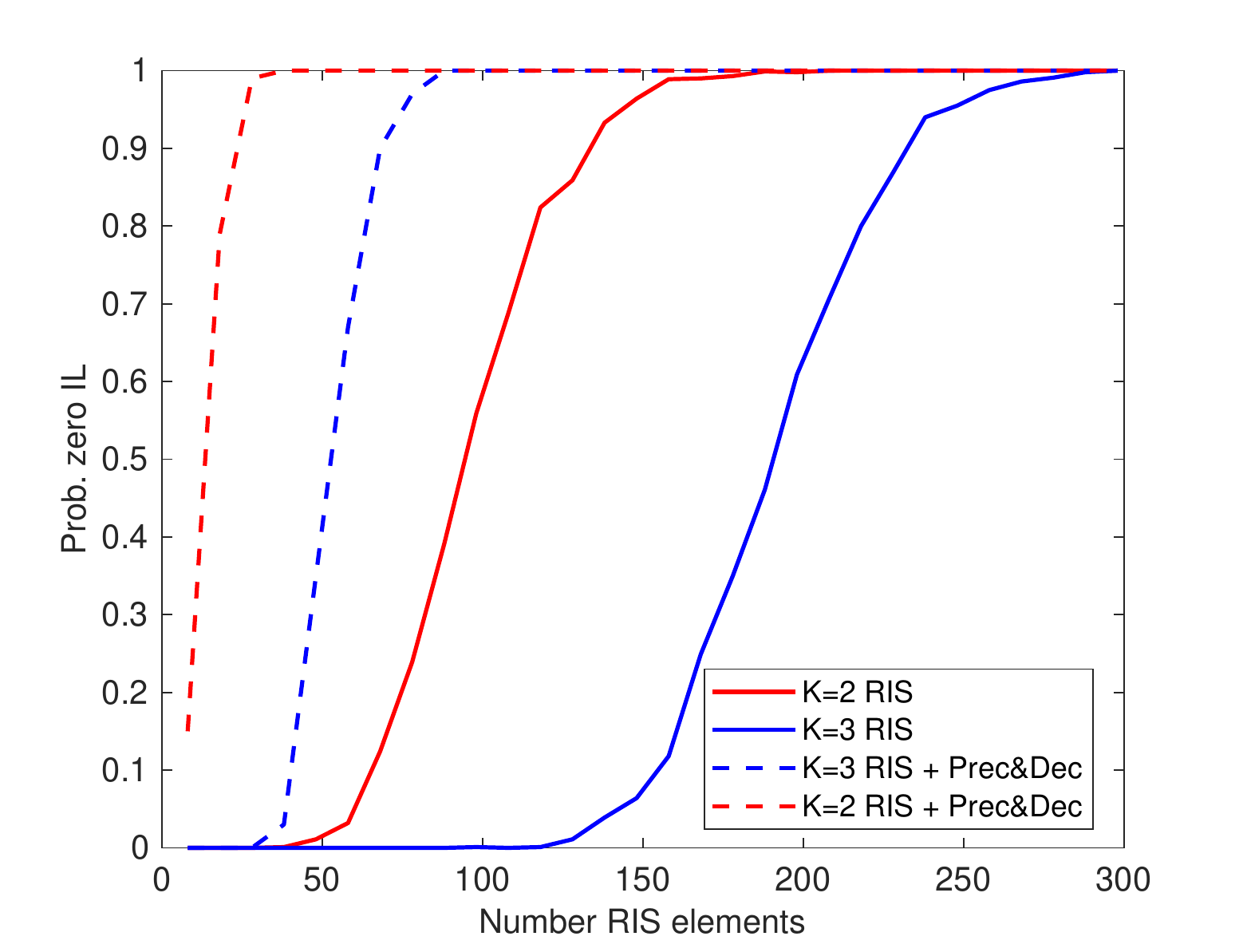}
     \caption{Probability that problem ${\cal P}_2$ (cf. \eqref{eq_zeroILproblem}) is feasible as a function of the number of RIS elements. }
	\label{fig:Probfeasible}
\end{figure}

Fig. \ref{fig:sum-rate} compares the sum-rates achieved in a RIS-assisted $(3\times3,2)^3$ MIMO IC by i) an active RIS with $M=25$ elements, and ii) a RIS with $M= 100$, $M=150$ and $M=300$ unit-modulus elements. In this example, the precoders and decoders are unitary random matrices. The curves have been obtained by averaging 1000 independent realizations of the MIMO IC. For this example, $g = \sum_{l \neq k} d_k d_l =24$, so an active RIS with $M=25$ achieves $IL=0$. However, when only the phases are optimized, nearly $M=300$ elements are needed to completely neutralize the interference and thus extract all the DoF of the IC. The sum-rate vs. SNR curves also implicitly reflect the IL performance of the different algorithms. Those cases where the sum-rate curve tends to saturate (cases with $M=100$ and $M=150$ in Fig. \ref{fig:sum-rate}) indicate that the IL does not converge to 0, the system attains 0 DoF, and IA is infeasible. Cases with $M=300$ or active RIS in Fig. \ref{fig:sum-rate} reach the maximum slope (maximum number of DoFs), thus implicitly indicating that the IL has converged to 0.

Fig. \ref{fig:Probfeasible} studies the probability that a passive lossless RIS can achieve zero IL (feasibility of problem ${\cal P}_2$ in  \eqref{eq_zeroILproblem}), as a function of the number of RIS elements. For each channel realization, if the IL after the optimization algorithm is less than $10^{-8}$, we declare the problem feasible; otherwise, we declare it infeasible. Fig. \ref{fig:Probfeasible} shows in solid line the case where only the RIS phases are optimized (random precoders and decoders are applied in this case) and in dashed line the case where the precoders, decoders, and RIS are optimized to minimize the IL. Clearly, when the precoders and decoders are also optimized, the number of RIS elements required to achieve zero-IL decreases significantly. When only the RIS phases are optimized, a RIS with $M \geq 150$ element achieves zero-IL with high probability for $K=2$, while for $K=3$ $M\geq 300$ elements are needed to completely cancel the interference.

\section{Conclusions}
The problem of achieving zero interference leakage (i.e., achieving perfect interference cancellation) in a RIS-assisted interference channel admits a geometrical interpretation that connects it with envelope precoding and phase-only zero-forcing problems. This, in turn, allows us to derive a set of necessary conditions for the feasibility of the problem. Our simulations show that, when the number of RIS elements is sufficiently high, it is possible with a high probability to perfectly cancel or neutralize the interference at all receivers.

\vfill
\pagebreak

\newpage
\bibliographystyle{IEEEbib}
\vfill\pagebreak
\balance
\bibliography{refs}

\end{document}